%% file: maxcut.tex
\documentclass[11pt]{article}

\usepackage{geometry}
\usepackage{lmodern}
\usepackage[T1]{fontenc}
\usepackage[mathletters]{ucs} 
\usepackage[utf8]{inputenc}

\usepackage{graphicx}
\graphicspath{{images/}}

\usepackage{amsmath, amsthm}
\usepackage{amssymb, amsbsy}
\usepackage[dvipsnames]{xcolor}
\usepackage{hyperref}
\hypersetup{
colorlinks=true,
linkcolor=Brown,
citecolor=OliveGreen,
urlcolor=RoyalBlue,
}

\usepackage{floatflt, wrapfig}
\usepackage{xcolor}

\usepackage{graphicx}
\graphicspath{ {images/} }
\usepackage{epstopdf}
\usepackage{color, import}
\usepackage{sectsty}
\usepackage{enumitem}
\usepackage{lineno}

\usepackage{cite}
\usepackage[edges]{forest}

\usepackage[linesnumbered,vlined,ruled]{algorithm2e}
\SetAlFnt{\footnotesize}
\SetAlCapFnt{\small}
\SetAlCapNameFnt{\small}

\usepackage{setspace}
\usepackage[list=true]{subcaption}
\usepackage[title]{appendix}
\usepackage{multirow}
\usepackage{array}

\usepackage{etoolbox}
\makeatletter
\patchcmd{\@maketitle}{\begin{center}}{\begin{flushleft}}{}{}
\patchcmd{\@maketitle}{\begin{tabular}[t]{c}}{\begin{tabular}[t]{@{}l}}{}{}
\patchcmd{\@maketitle}{\end{center}}{\end{flushleft}}{}{}
\makeatother

\renewenvironment{abstract}
{\small\section*
{\bfseries\noindent{\raisebox{-.15\baselineskip}{\normalsize\abstractname}}\hrulefill} 
}

\newtheorem{theorem}{Theorem}
\newtheorem{lemma}{Lemma}

\author{hkhj\\hjk}
\author{hjjkhj\\hjk}

\begin{document}
\title{Complexity of Maximum Cut on Interval Graphs}
    \author{
    Ranendu Adhikary\\
    \small \emph{Department of Mathematics, Jadavpur University, India}\\
    \small \emph{ranenduadhikary.rs@jadavpuruniversity.in}\\\\
    Kaustav Bose\\
    \small \emph{Department of Mathematics, Jadavpur University, India}\\
    \small \emph{kaustavbose.rs@jadavpuruniversity.in}\\\\
    Satwik Mukherjee\\
    \small \emph{Department of Mathematics, Jadavpur University, India}\\
    \small \emph{m.satwik90@gmail.com}\\\\
    Bodhayan Roy\\
    \small \emph{Department of Mathematics, Indian Institute of Technology Kharagpur, India}\\
    \small \emph{broy@maths.iitkgp.ac.in}\\\\
    }
    \date{}
    
  \maketitle

\begin{abstract}
%

We resolve the longstanding open problem concerning the computational complexity of \textsc{Max Cut} on interval graphs by showing that it is NP-complete.

\end{abstract}
\section{Introduction}

For a graph $G = (V, E)$, a \emph{cut} is a partition of $V$ into two disjoint subsets. Any cut determines a \emph{cut set} which is the set of all edges that have one endpoint in one partition and the other endpoint in the other partition. The \emph{size} of a cut is the cardinality of its cut set. The maximum cut problem or \textsc{Max Cut} asks for a cut of maximum size. \textsc{Max Cut} is a fundamental and well-known NP-complete problem \cite{garey1990guide}. The weighted version of the problem is one of Karp's original 21 NP-complete problems \cite{karp1972reducibility}. Besides its theoretical importance, it has applications in VLSI circuit design \cite{chang1987}, statistical physics \cite{barahona1988application} etc. \textsc{Max Cut} remains NP-hard even for cubic graphs \cite{berman1999some}, split graphs \cite{bodlaender2000complexity}, co-bipartite graphs \cite{bodlaender2000complexity}, unit disk graphs \cite{diaz2007max} and total graphs \cite{guruswami1999maximum}. On the positive side, polynomial time algorithms are known for planar graphs \cite{hadlock1975finding}, line graphs \cite{guruswami1999maximum}, graphs not contractible to $K_5$ \cite{barahona1983max} and graphs with bounded treewidth \cite{bodlaender2000complexity}.

It is well known that many classical NP-complete problems like colourability \cite{golumbic2004}, Hamiltonian cycle \cite{keil1985finding}, minimum dominating set \cite{chang1998efficient}, minimum feedback vertex set \cite{lu1997linear}, minimum vertex cover \cite{marathe1992generalized} and maximum clique \cite{imai1983finding} are polynomial time  solvable for interval graphs. This is because interval graphs are well structured graphs with many nice properties and decomposition models that are often exploited to design efficient dynamic programming or greedy algorithms. Few problems that are known to be NP-hard in interval graphs include optimal linear arrangement \cite{cohen2006optimal}, achromatic number \cite{bodlaender1989achromatic}, harmonious colouring \cite{asdre2007harmonious}, geodetic set \cite{chakraborty2020algorithms}, minimum sum colouring \cite{marx2005short}, metric dimension \cite{foucaud2017}, identifying code \cite{foucaud2017} and locating-dominating set \cite{foucaud2017}. The class of interval graphs is widely regarded as an important graph class with many real-world applications. Interval graphs arise naturally in modelling problems that involve temporal reasoning, e.g scheduling problems. Interval graphs are also extensively used in bioinformatics (e.g. DNA mapping \cite{zhang1994algorithm}, protein sequencing \cite{jungck1982computer}) and mathematical biology (e.g. food webs in population biology \cite{cohen1978food}).

Surprisingly, the computational complexity of \textsc{Max Cut} for interval graphs is a longstanding open problem. The first time that the problem was mentioned as open was probably in 1985 \cite{JOHNSON1985434}. No polynomial time algorithm is known even for the subclass of unit interval graphs. There are two previous works \cite{bodlaender1999simple, boyaci2017polynomial} reporting polynomial time algorithms solving \textsc{Max Cut} for unit interval graphs. However, both algorithms were later reported to be incorrect \cite{bodlaender2004simple, kratochvil2020u}. In this paper, we show that \textsc{Max Cut} is NP-complete for interval graphs.

\section{Preliminaries}

For any simple undirected graph $G = (V, E)$, a \emph{cut} is a partition of $V$ into two disjoint subsets $A$ and $B$, i.e., $V = A \cup B$ and $A \cap B = \emptyset$. The corresponding \emph{cut set} is the set of all edges that have one endpoint in $A$ and the other endpoint in $B$, i.e., the set $\{(u,v) \in E \mid (u \in A, v \in B) \vee (u \in B, v \in A) \}$. The \emph{size} of the cut is the cardinality of its cut set. A typical instance of the decision version of \textsc{Max Cut} consists of a simple undirected graph $G = (V, E)$ and an integer $k$ such that $1 \leq k \leq |E|$. $(G, k)$ is an yes-instance of \textsc{Max Cut} if and only if $G$ has a cut of size at least $k$. 

Interval graphs are the intersection graphs of intervals on the real line. Formally, $G = (V, E)$ is said to be an \emph{interval graph} if there is a set $S$ of intervals on the real line and a bijection $\varphi : V \longrightarrow S$ such that $u, v \in V$ are adjacent if and only if $\varphi(u) \cap \varphi(v) \neq \emptyset$.

\section{NP-Completeness} 
In this section, we show that \textsc{Max Cut} is NP-complete on interval graphs. \textsc{Max Cut} is known to be NP-complete on cubic graphs \cite{berman1999some}. We reduce \textsc{Max Cut} on cubic graphs to \textsc{Max Cut} on interval graphs.

\subsection{Construction of the Reduction Graph} \label{sec: reduction}
%
%
%
%
Let $(G, x)$ be an instance of \textsc{Max Cut} where $G = (V, E)$ is a cubic graph. Let $|V| = n$ and hence $|E| = \frac{3}{2}n$. We shall reduce it to an equivalent instance $(G', f(x))$ of \textsc{Max Cut} where $G' = (V', E')$ is an interval graph. The construction of $G'$ is outlined in the following. $G' = (V', E')$ is described as the intersection graph of a set of intervals on the real line and the vertices of $G'$ are referred to as intervals.

\begin{enumerate}
 \item Fix an arbitrary ordering of the vertices and edges of $G = (V, E)$ as $v_1, v_2, \ldots, v_n, e_1, e_2,$ $\ldots, e_m$. We shall write any edge $e \in E$ as an ordered pair of vertices that respects the following convention. If $e$ is an edge between $v_i$ and $v_j$, where $i < j$, then we shall write $e = (v_i, v_j)$ (not $e = (v_j, v_i)$).
 
 
 \item For each vertex $v \in V$, we construct a V-gadget $\mathcal{G}(v)$ and for each edge $e \in E$, we construct an E-gadget $\mathcal{G}(e)$. They are shown in Fig. \ref{fig: gadgets}. The structure of a V-gadget is identical to that of an E-gadget, the only difference is their size. Each V-gadget (E-gadget) consists of $q$ (resp. $q'$) left long intervals, $p$ (resp. $p'$) left short intervals, $q$ (resp. $q'$) right long intervals and $p$ (resp. $p'$) right short intervals. The left long intervals and the right long intervals of a V-gadget (E-gadget) all intersect each other to form a clique of size $2q$ (resp. $2q'$). All left short intervals of a V-gadget (E-gadget) are mutually disjoint and each of them intersect only the $q$ (resp. $q'$) left long intervals. Similarly all right short intervals of a V-gadget (E-gadget) are mutually disjoint and each of them intersect only the $q$ (resp. $q'$) right long intervals. Therefore, the number of edges in each V-gadget (E-gadget) is $q(2q-1) + 2pq$ (resp.  $q'(2q'-1) + 2p'q'$). 
 

   \begin{figure}[h!]
\centering
\subcaptionbox[Short Subcaption]{
        \label{1}
}
[
    0.46\textwidth 
]
{
    \fontsize{8pt}{8pt}\selectfont
    \def\svgwidth{0.46\textwidth}
    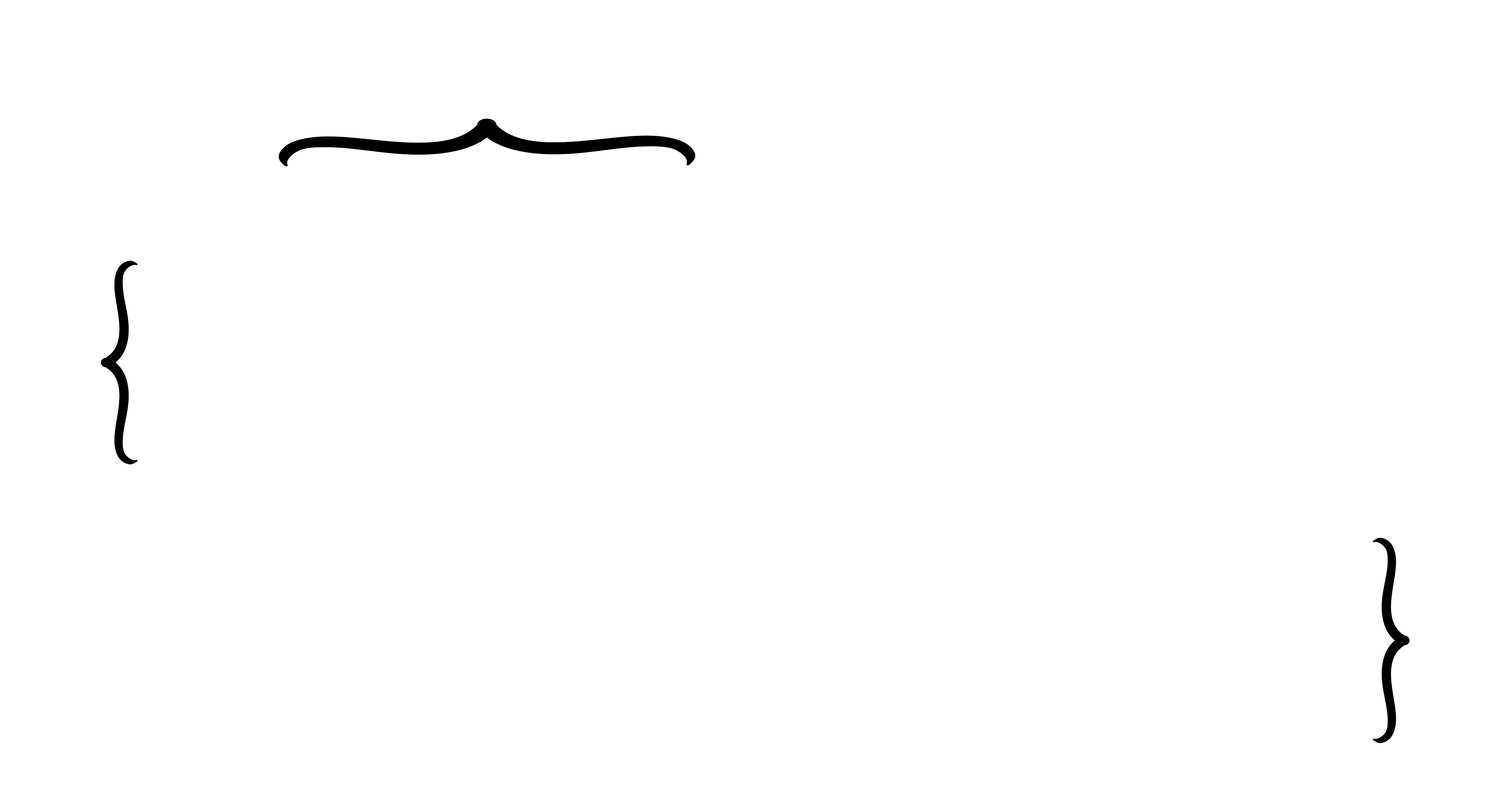
}
\hfill 
\subcaptionbox[Short Subcaption]{
     \label{2}
}
[
    0.46\textwidth 
]
{
    \fontsize{8pt}{8pt}\selectfont
    \def\svgwidth{0.46\textwidth}
    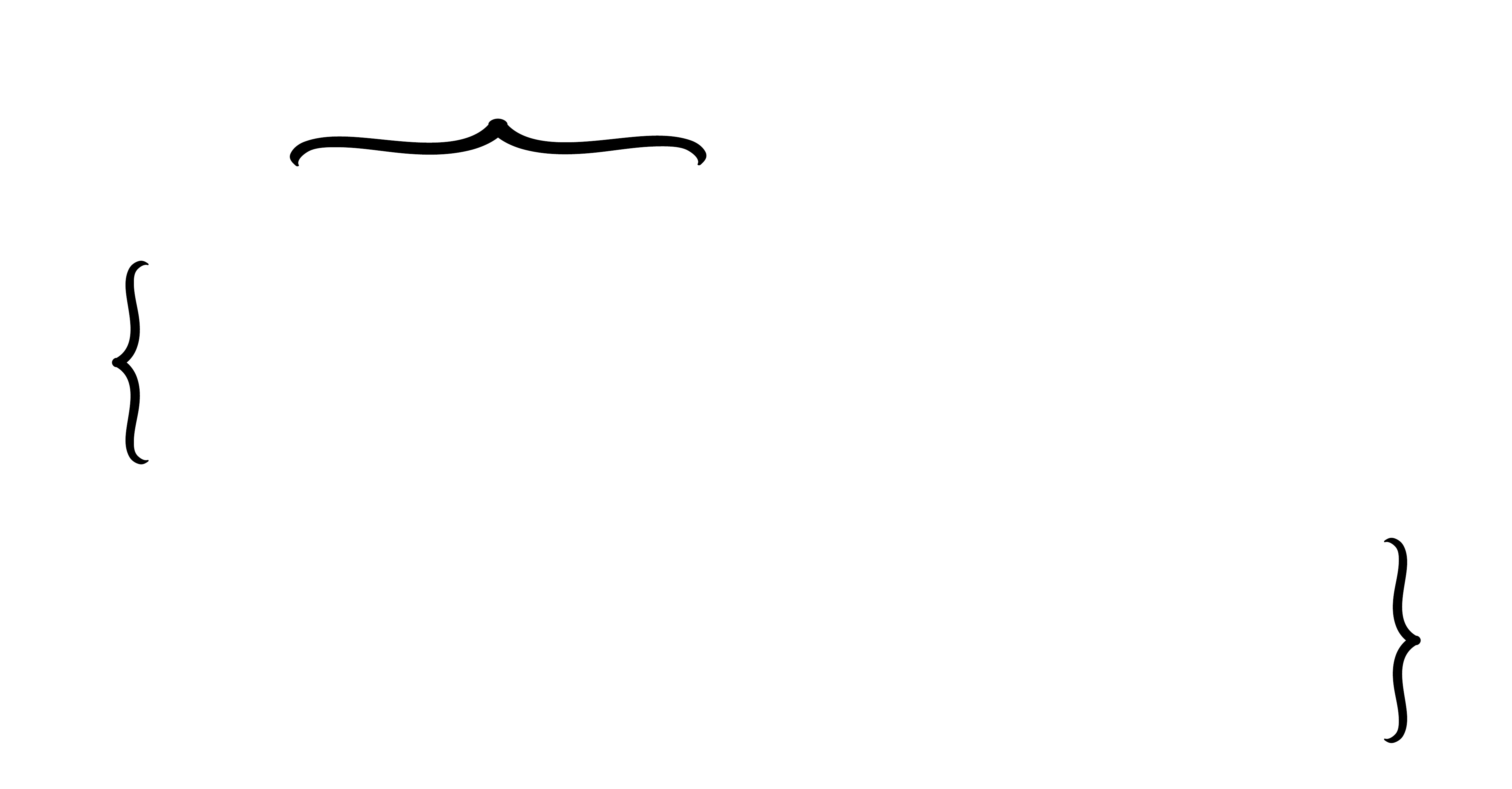
}

\caption[Short Caption]{a)  A V-gadget. b)  An E-gadget.}
\label{fig: gadgets}
\end{figure}

\item
 We set $q = 200n^3 + 1$, $p=2q+7n$, $q' = 10n^2 + 1$, $p' = 2q'+7n$, where $n$ is the number of vertices in $G$. Note that the following inequalities hold:
 \begin{enumerate} 
  \item $p>2q>2p'>4q'> 9n^2$
  \item $q^2>(p-q)6n$
  \item $q'^2>(p'-q')6n$
  \item $q > 3(p' + q')n + 9n^2$

  \end {enumerate}

\item There are a total of $n$ V-gadgets, and $3n/2$ E-gadgets. 
All $5n/2$ gadgets are arranged in the following order as shown in Fig. \ref{fig: arrangement} : $\mathcal{G}(v_1), \mathcal{G}(v_2), \ldots, \mathcal{G}(v_n),$ $\mathcal{G}(e_1), \mathcal{G}(e_2), \ldots, \mathcal{G}(e_{3n/2})$. No two intervals belonging to different gadgets intersect.

 \item To establish relationships between the V-gadgets and E-gadgets we introduce $6n$ link intervals (See Fig. \ref{fig: arrangement}). Link intervals connect V-gadgets to E-gadgets. This will be described in the next point. A link interval can intersect a gadget in four different ways as described in the following.
 
 \begin{itemize}
  \item A link interval is said to \emph{cover a gadget} if it intersects all intervals of the gadget. (See Fig. \ref{cover})
  
  \item A link interval is said to \emph{intersect a V-gadget in the first manner} if it intersects only the $q$ right long intervals of the V-gadget. (See Fig. \ref{intersect1}). 
  
  \item A link interval is said to \emph{intersect an E-gadget in the second manner} if it intersects only the $p'$ left long intervals of the gadget. (See Fig. \ref{intersect2}). 
  
  \item A link interval is said to \emph{intersect an E-gadget in the third manner} if it intersects only the $q'$ left long intervals and the $p'$ left short intervals of the gadget. (See Fig. \ref{intersect3}). 
 \end{itemize}

 \item For each edge $e = (v_i, v_j) \in E$, we introduce four link intervals: 1) a pair intersecting $\mathcal{G}(v_i)$ in the first manner and $\mathcal{G}(e)$ in the second manner, and 2) another pair intersecting $\mathcal{G}(v_j)$ in the first manner and $\mathcal{G}(e)$ in the third manner  (See Fig. \ref{fig: connection}). Note that since $G$ is cubic, the total number of link intervals covering a V-gadget is $6k$ for some integer $k$, where $k$ may vary from $0$ to $n-1$. Similarly, the total number of link intervals covering an E-gadget is $4k$ for some integer $k$, where $k$ may vary from $0$ to $3n/2-1$. Also, the total number of link intervals intersecting a V-gadget in the first manner is $6$.
 \end{enumerate}


\begin{figure}[h!]
\centering

{
    \fontsize{8pt}{8pt}\selectfont
    \def\svgwidth{1\textwidth}
    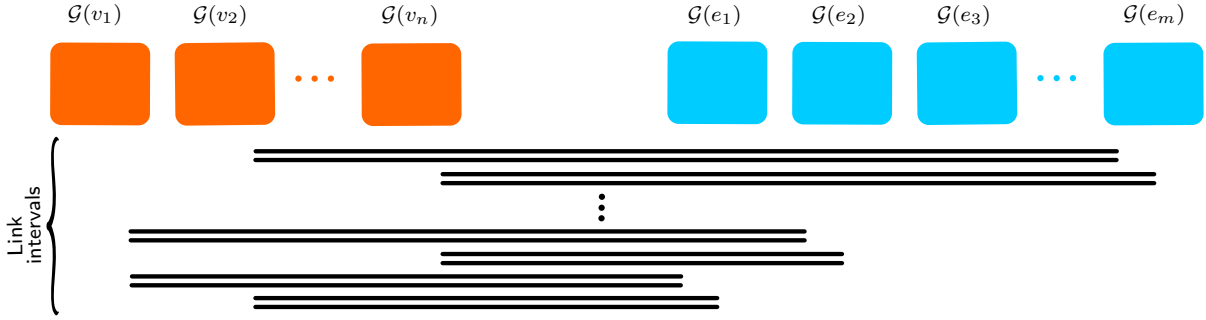
}

\caption[Short Caption]{Arrangement of the gadgets and the link intervals.}
\label{fig: arrangement}
\end{figure}

   \begin{figure}[h!]
\centering
\subcaptionbox[Short Subcaption]{A gadget is covered by a link interval.
        \label{cover}
}
[
    0.4\textwidth 
]
{
    \fontsize{8pt}{8pt}\selectfont
    \def\svgwidth{0.4\textwidth}
    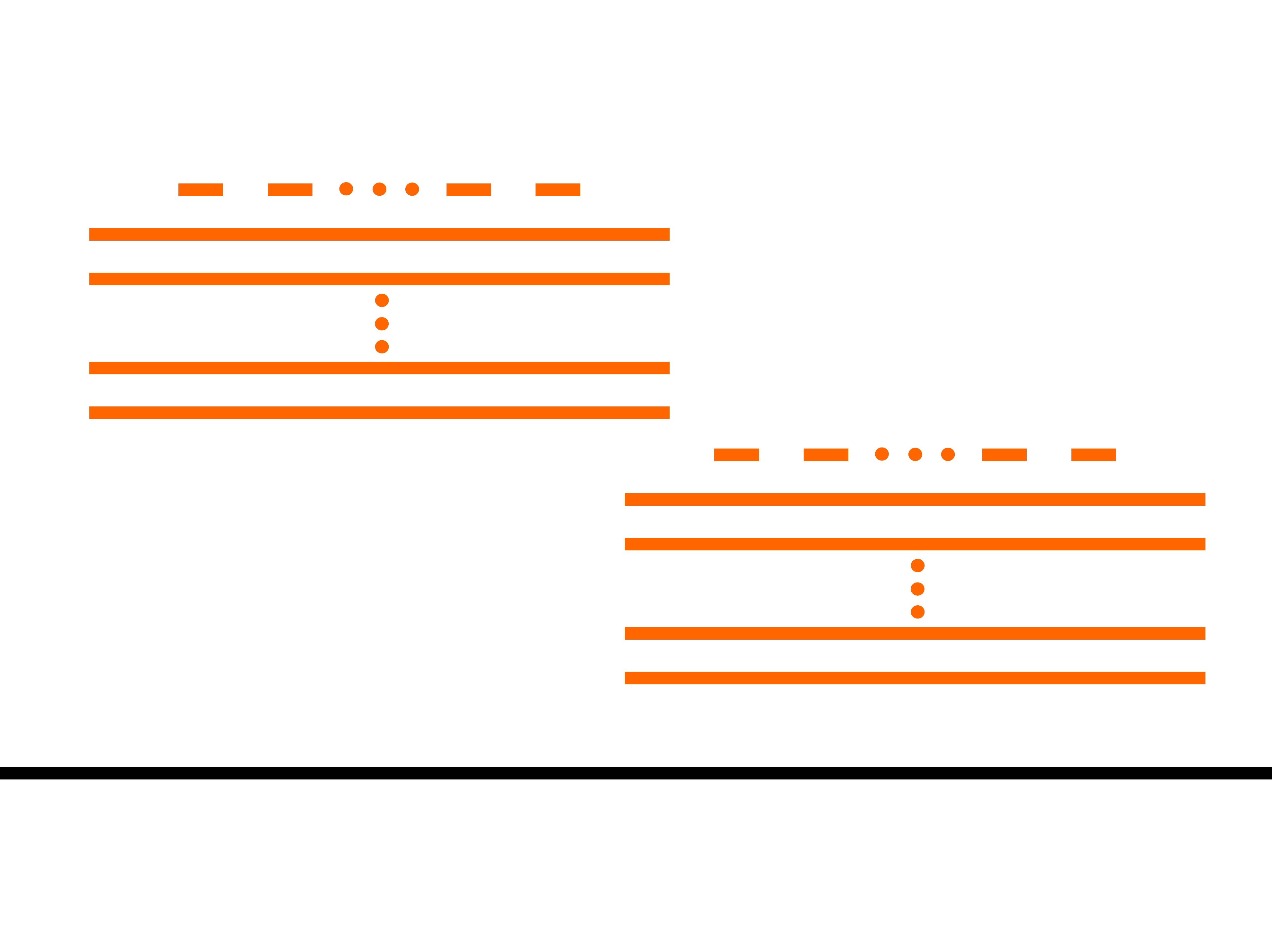
}
\hfill 
\subcaptionbox[Short Subcaption]{A link interval intersects a V-gadget in the first manner.
     \label{intersect1}
}
[
    0.4\textwidth 
]
{
    \fontsize{8pt}{8pt}\selectfont
    \def\svgwidth{0.4\textwidth}
    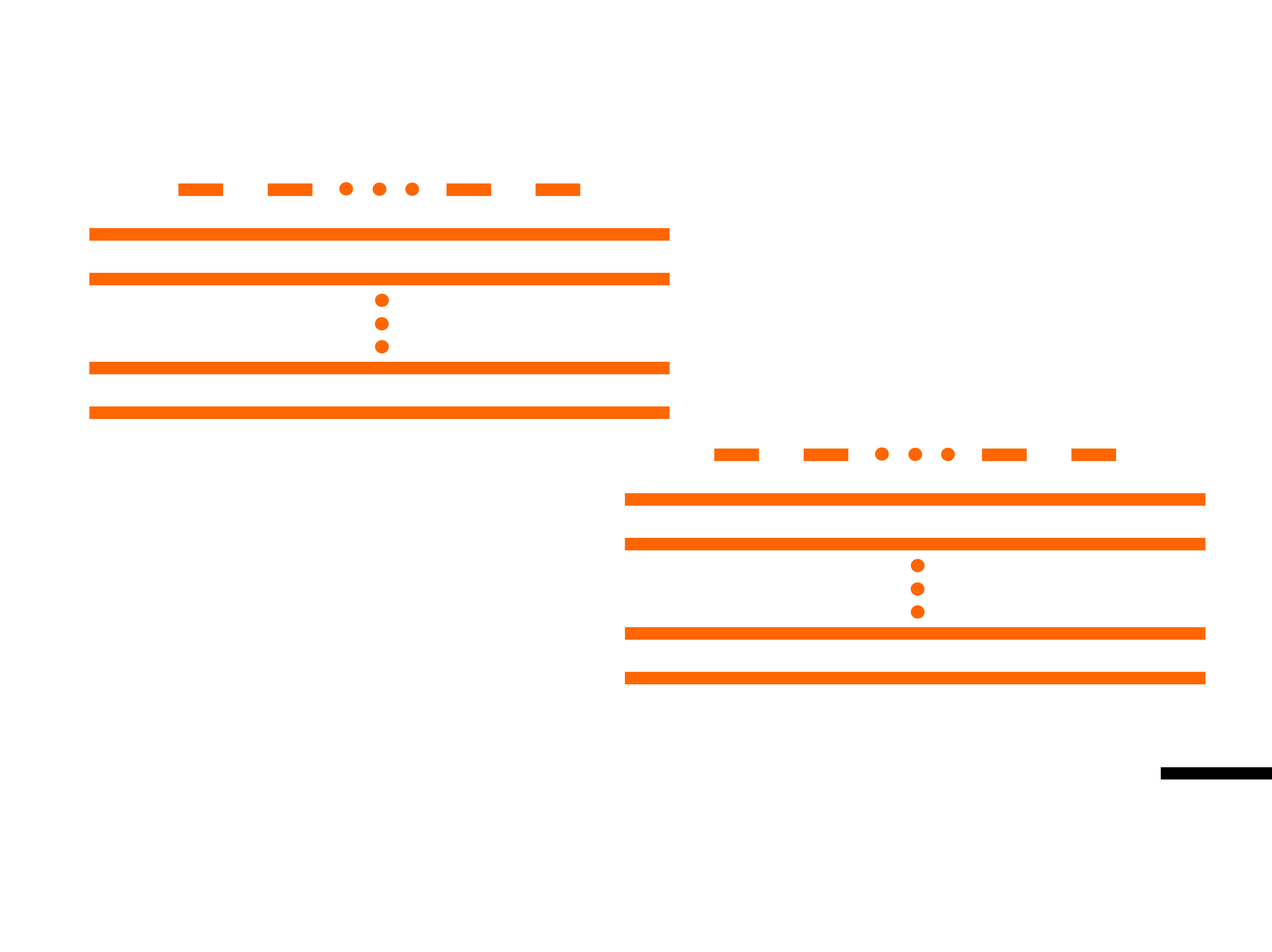
}
\\
\subcaptionbox[Short Subcaption]{A link interval intersects an E-gadget in the second manner.
         \label{intersect2}
}
[
    0.4\textwidth 
]
{
    \fontsize{8pt}{8pt}\selectfont
    \def\svgwidth{0.4\textwidth}
    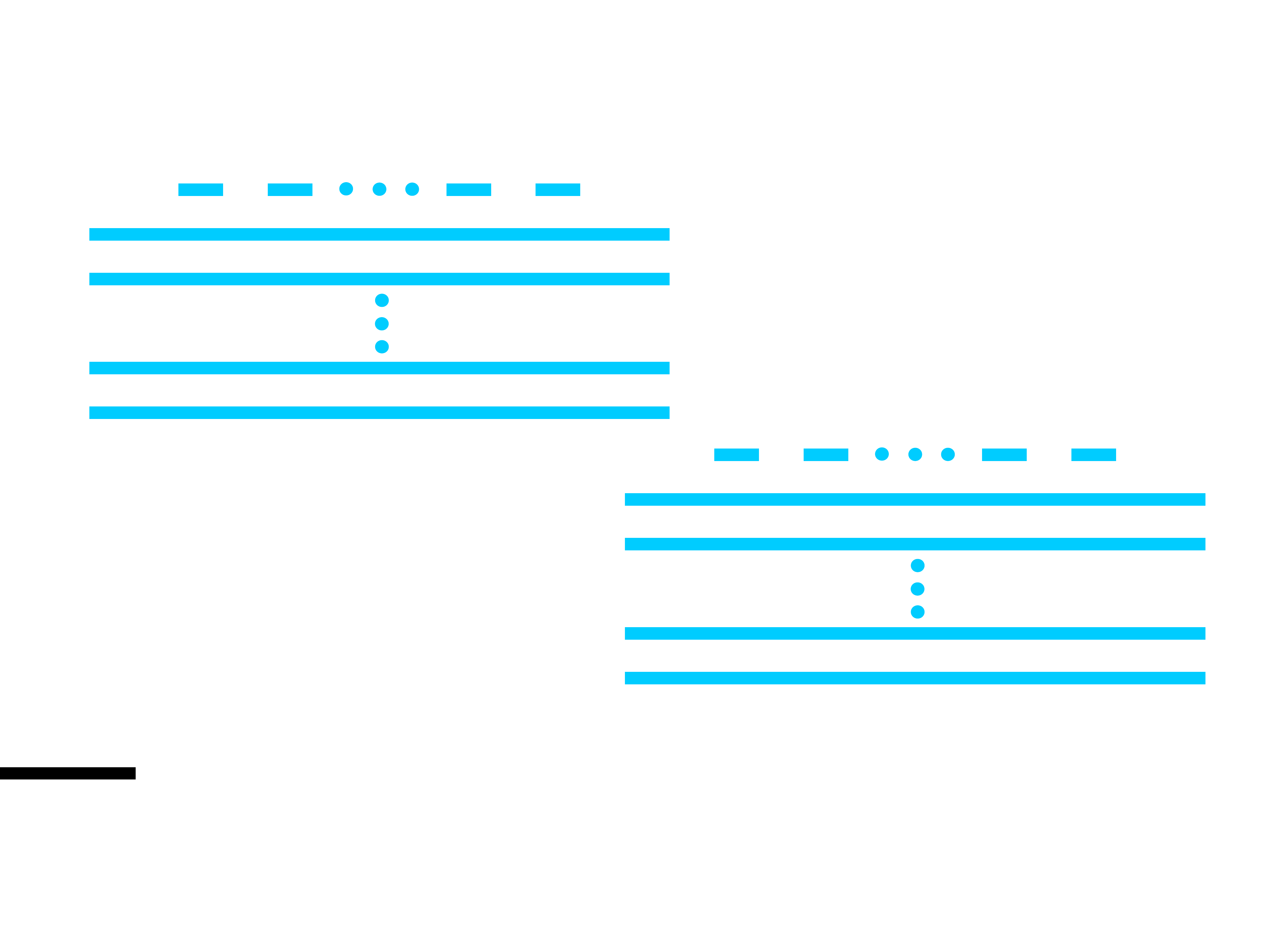
}
\hfill
\subcaptionbox[Short Subcaption]{A link interval intersects an E-gadget in the third manner.
        \label{intersect3}
}
[
    0.4\textwidth 
]
{
    \fontsize{8pt}{8pt}\selectfont
    \def\svgwidth{0.4\textwidth}
    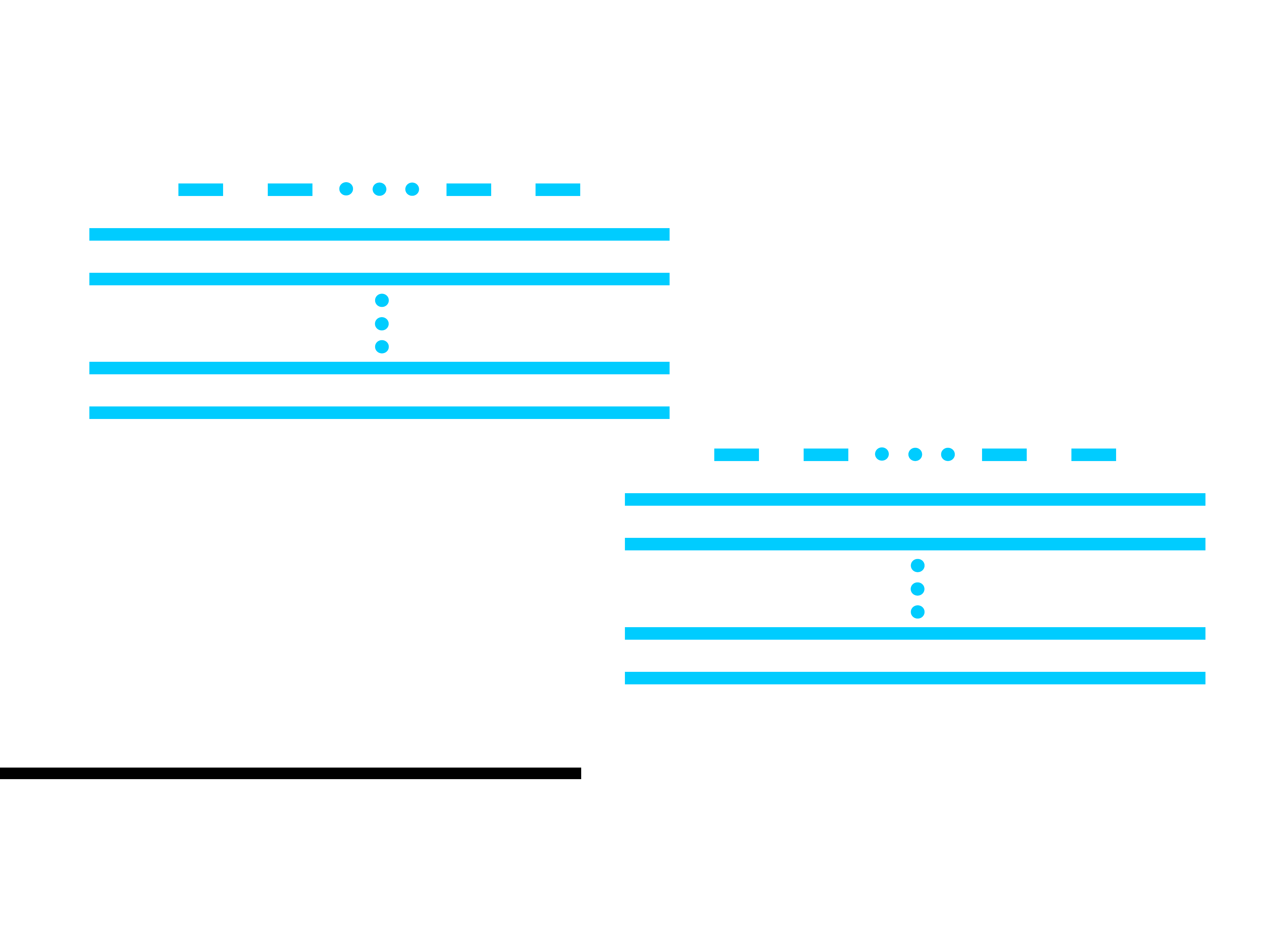
}

\caption[Short Caption]{Illustrations showing the four different ways a link interval can intersect a gadget. }
\label{notcollision}
\end{figure}

\begin{figure}[h!]
\centering

{
    \fontsize{8pt}{8pt}\selectfont
    \def\svgwidth{1\textwidth}
    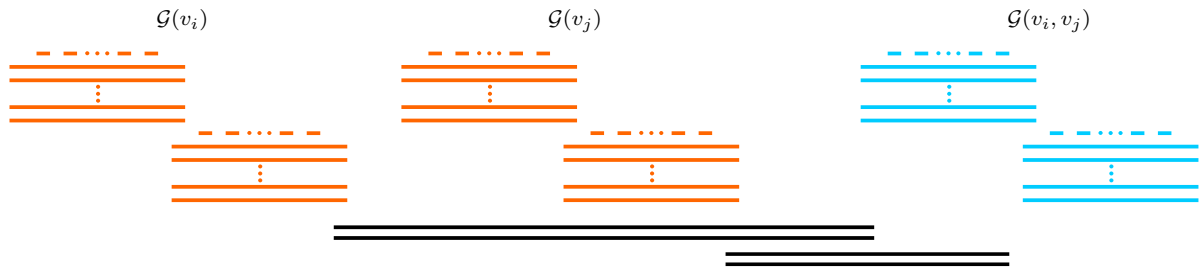
}

\caption[Short Caption]{link intervals connecting an E-gadget $\mathcal{G}((v_i, v_j))$ with V-gadgets $\mathcal{G}(v_i)$ and $\mathcal{G}(v_j)$.  }
\label{fig: connection}
\end{figure}

\subsection{Properties of the Reduction Graph} \label{sec: proof}

In this section, we study some properties of the interval graph $G'$ constructed from $G$ in the previous section. We consider a partition of vertices of $G'$ that yields a maximum cut. To prove that the partition satisfies some properties, in general we show that if it does not satisfy those properties, then the size of the corresponding cut can be increased, contradicting the maximality of the cut.


Now consider a maximum cut of $G'$ with the vertices partitioned into subsets $A$ and $B$. We show in the next lemma that for every vertex gadget $\mathcal{G}(v_i)$, either $A$ or $B$ contains all of its left short intervals. The same holds for the right short intervals of $\mathcal{G}(v_i)$.

\begin{lemma} \label{lem1}
 If a partition of $G'$ yields a maximum cut, then for any V-gadget $\mathcal{G}(v_i)$, all of its left short intervals lie in the same subset. The same holds for its right short intervals.
\end{lemma}

\begin{proof}
Consider a maximum cut of $G'$ that partitions its vertices into subsets $A$ and $B$. Let $LL^A_i$ and $LL^B_i$ denote the subset of left long intervals of $\mathcal{G}(v_i)$ in $A$ and $B$ respectively. Denote by $OL^A_i$ (resp. $OL^B_i$) the set of all link intervals that cover $\mathcal{G}(v_i)$ and lie in subset $A$ (resp. $B$). Without loss of generality, let the following direction of inequality hold: 
 $$\mid LL^A_i \mid + \mid OL^A_i \mid > \mid LL^B_i \mid + \mid OL^B_i \mid$$
 Note that the inequality must be strict since the sum of the number of left long intervals of $\mathcal{G}(v_i)$ and the number of  link intervals covering $\mathcal{G}(v_i)$ is $6k + q$ ($0 \leq k \leq n-1$), which is odd since $q$ is odd. Suppose that a left short interval of $\mathcal{G}(v_i)$ is in $A$. Recall that the left long intervals of $\mathcal{G}(v_i)$ and the link intervals covering $\mathcal{G}(v_i)$ are the only intervals that a left short interval of $\mathcal{G}(v_i)$ intersects. Then due to the above inequality, moving the left short interval to $B$ increases the number of cut edges. This contradicts the fact that the partition yields a maximum cut. Hence, all left short intervals of $\mathcal{G}(v_i)$ must be in $B$. Using similar arguments we can show that all right short intervals of $\mathcal{G}(v_i)$ must be same subset.
\end{proof}


%
%
%
%
%

In the following lemma, we prove a property of the long intervals of each vertex gadget, akin to the property of the short intervals proved in the previous lemma.

\begin{lemma} \label{lem2}
If a partition of $G'$ yields a maximum cut, then for any V-gadget $\mathcal{G}(v_i)$, all its left long intervals lie in the same subset. The same holds for its right long intervals.
%
\end{lemma}

\begin{proof}
 Consider a maximum cut of $G'$ that partitions its vertices into subsets $A$ and $B$. By Lemma \ref{lem1}, all of the short intervals on the same side of $\mathcal{G}(v_i)$ belong to the same subset. Without loss of generality, we consider two cases, where (a) all the left short intervals of $\mathcal{G}(v_i)$ are in $A$, and all the right short intervals of $\mathcal{G}(v_i)$ are in $B$, and (b) all the short intervals of $\mathcal{G}(v_i)$ are in $A$.
 
First consider Case (a), where all the left short intervals of $\mathcal{G}(v_i)$ belong to $A$, and all the right short intervals of $\mathcal{G}(v_i)$ belong to $B$. Suppose that a left long interval of $\mathcal{G}(v_i)$ is in $A$. Then moving it to $B$ results in losing at most $2q-1$ cut edges due to its intersections with other long intervals of $\mathcal{G}(v_i)$, and at most $6n$ cut edges due to its intersections with the link intervals of $\mathcal{G}(v_i)$.
However, we gain at most $p$ cut edges. Since $p=2q+7n$, the quantity $p-(2q-1+6n)$ is positive, hence the size of the cut increases. This contradicts the fact that the partition yields a maximum cut. Hence, all left long intervals of $\mathcal{G}(v_i)$ must be in $B$.

Now consider Case (b), where all the short intervals of $\mathcal{G}(v_i)$ belong to $A$. It can be seen that the above argument is also applicable in this case, and the claim holds. 
\end{proof}

In the following lemma, we consolidate the results obtained above into a complete partition of a vertex gadget in a maximum cut.

\begin{lemma} \label{lem3}
If a partition of $G'$ yields a maximum cut, then for any V-gadget $\mathcal{G}(v_i)$, all the left long and right short intervals are in one subset, while all the right long and left short intervals are in the other.

\end{lemma}

\begin{proof}
 Consider a maximum cut of $G'$ that partitions its vertices into subsets $A$ and $B$. Then without loss of generality, by Lemma \ref{lem1}, either (a) all left short intervals of $\mathcal{G}(v_i)$ are in $A$ and all right short intervals of $\mathcal{G}(v_i)$ are in $B$, or (b) all the short intervals of $\mathcal{G}(v_i)$ are in $A$. 
 
 First consider Case (a), i.e. $\mathcal{G}(v_i)$ has all its left short intervals in $A$ and right short intervals in $B$. Then it follows from the proof of Lemma \ref{lem2} that all left long intervals of $\mathcal{G}(v_i)$ are  in $B$ and all right long intervals of $\mathcal{G}(v_i)$ must be in $A$, as claimed. 
  
  Now consider Case (b), i.e. $\mathcal{G}(v_i)$ has all its short intervals in $A$.  Since all the short intervals of $\mathcal{G}(v_i)$ are in $A$, it implies from the proof of \ref{lem2} that  all the long intervals of $\mathcal{G}(v_i)$ are in $B$.  We move all the right short intervals of $\mathcal{G}(v_i)$ to $B$ and all right long intervals of $\mathcal{G}(v_i)$ to $A$. Due to their intersections with link intervals, this removes at most $(p-q)6n$ edges from the cut. But due to the intersections among the left and right long intervals, it also adds at least $q^2$ edges to the cut. Since by our choice of $q$ and $p$, we have $q^2 - (p-q)6n > 0$, the total number of edges in the cut increases. This contradicts the fact that the partition yields a maximum cut and hence this case is impossible.
\end{proof}

It can be seen that V-gadgets and E-gadgets are structurally similar, and only their intersections with the link intervals can possibly be the cause of any different partitioning in a maximum cut. We address this point in the following lemma and show that E-gadgets too in fact admit a partition similar to that of V-gadgets.
 
\begin{lemma} \label{lem5}
 Lemma \ref{lem3} holds for E-gadgets of $G'$ as well.
\end{lemma}

\begin{proof}

Consider a maximum cut of $G'$ that partitions its vertices into subsets $A$ and $B$. We modify the proof of Lemma \ref{lem1} a little, so that Lemma \ref{lem1} holds for E-gadgets as well. Consider an E-gadget $\mathcal{G}(e_i)$ of $G'$. Observe that the proof holds for the right short intervals of $\mathcal{G}(e_i)$, since any link interval that intersects the right short intervals of an E-gadget, must also cover the E-gadget. But the left short intervals of each E-gadget are intersected by two link intervals in the third manner. Then denote by $OL'^A_i$ (resp. $OL'^B_i$) the set of all link intervals that cover $\mathcal{G}(v_i)$ or intersect $\mathcal{G}(v_i)$ in the third manner, and lie in subset $A$ (resp. $B$). Let $LL^A_i$ and $LL^B_i$ denote the subset of left long intervals of $\mathcal{G}(v_i)$ in $A$ and $B$ respectively, as before. Again, without loss of generality we have the following inequality.
 $$\mid LL^A_i \mid + \mid OL'^A_i \mid > \mid LL^B_i \mid + \mid OL'^B_i \mid$$
 The rest of the proof is similar to that of \ref{lem1}, and it can be seen that the claim holds. The proof of Lemma \ref{lem2} for E-gadgets remains the same as for V-gadgets. Lemmas \ref{lem1} and \ref{lem2} along with the choice of $p'$ and $q'$, imply Lemma \ref{lem3} for E-gadgets as well. 
\end{proof}

\begin{lemma} \label{lem6}
 $G$ has a cut of size at least $x$ if and only if $G'$ has a cut of size at least $(2pq + q^2)n +  \frac{3}{2}(2p'q' + q'^2)n +  3(n-1)(n-2)(p+q) + 3n(\frac{3}{2}n-1) (p'+q') + 6nq + 3np' + 2xq'$.
\end{lemma}

\begin{proof}
First suppose that $G$ has a cut of size at least $x$. Denote the subsets in the partition of the vertices of $G$ by $C$ and $D$. We partition the vertices of $G'$ as follows. If a vertex $v_i$ of $G$ is in $C$, then in the corresponding V-gadget $\mathcal{G}(v_i)$ of $G'$, all left short intervals and right long intervals are placed in $A$, all right short intervals and left long intervals are placed in $B$. Finally, all link intervals intersecting $\mathcal{G}(v_i)$ in the first manner are placed in $B$. If $v_i$ is in $D$ instead, then all the above placements of intervals are swapped. Recall that for each E-gadget exactly two link intervals intersect it in the second manner and exactly two link intervals intersect it in the third manner. If the link intervals that intersect an E-gadget in the third manner is in $A$, then we place the left short intervals and right long intervals of the E-gadget in $B$, and the left long intervals and right short intervals in $A$. If the link intervals are in $B$, then the placements of the intervals are swapped. 

Due to the above placement of intervals in $A$ and $B$, the number of cut edges obtained internally from all the V-gadgets and E-gadgets of $G'$ are $(2pq + q^2)n$ and $\frac{3}{2}(2p'q' + q'^2)n$ respectively. The number of cut edges formed by the V-gadgets and the link intervals that cover them is $3(n-1)(n-2)(p+q)$. The number of cut edges formed by the E-gadgets and the link intervals covering them is $3n(\frac{3}{2}n-1) (p'+q')$. For each V-gadget, the link intervals intersecting it in the first manner give $6q$ cut edges, resulting in a total of $6nq$ cut edges. Each link interval that intersects an E-gadget in the third manner gives $p'$ cut edges, thus we have $3np'$ in total. However, a link interval that intersects an E-gadget in the second manner can produce cut edges from the E-gadget only when the other link interval mentioned above is in a different subset, i.e. the vertices of $G$ corresponding to the V-gadgets of these  link intervals are in $C$ and $D$, and produce a cut edge. This means that such link intervals produce at least $2xq'$ cut edges in total, proving the forward direction of the claim.

Now we prove the backward direction of the claim. Assume that $G'$ has a cut of size at least $(2pq + q^2)n +  \frac{3}{2}(2p'q' + q'^2)n +  3(n-1)(n-2)(p+q) + 3n(\frac{3}{2}n-1) (p'+q') + 6nq + 3np' + 2xq'$. So the size of a maximum cut of $G'$ is at least this much. Consider a maximum cut of $G'$ that partitions its intervals into two disjoint subsets $A$ and $B$. By Lemma \ref{lem3}, for each V-gadget, all the left long and right short intervals are in one subset, while all the right long and left short intervals are in the other. Corresponding to this cut of $G'$, we define a cut of $G$ in the following way. If the left long and right short intervals of  $\mathcal{G}(v_i)$ are in $A$ (resp. $B$), then we put $v_i$ in $C$ (resp. $D$). Let $y$ be the size of the cut $C \cup D$. We have to show that $y \geq x$.  

Due to Lemma \ref{lem3} and \ref{lem5}, the internal cut edges of V-gadgets and E-gadgets, and the cut edges formed between gadgets and the link intervals that cover them amount to $(2pq + q^2)n +  \frac{3}{2}(2p'q' + q'^2)n +  3(n-1)(n-2)(p+q) + 3n(\frac{3}{2}n-1)(p'+q')$ cut edges in total. Hence, the remaining $6nq + 3np' + 2xq'$ cut edges are obtained from the partial intersections of the link intervals with the V-gadgets and E-gadgets, and the intersections among link intervals. The number of cut edges among the link intervals is not more than $(3n)^2 = 9n^2$. The partial intersections between link intervals and V-gadgets can contribute at most $6nq$ cut edges. Note that the partial intersections between link intervals and E-gadgets, and intersections among the link intervals cannot give more than $3(p' + q')n + 9n^2$ cut edges. Since $q > 3(p' + q')n + 9n^2$, it implies that exactly $6nq$ of the remaining cut edges are obtained from  link intervals intersecting V-gadgets in the first manner. This happens when for each V-gadget, the link intervals intersecting it in the first manner are all in the subset which contains the left long and right short intervals of the gadget. Hence, the placement of the intervals of the V-gadget in the subsets $A$ and $B$ (and hence the placement of the corresponding vertex of $G$ in $C$ or $D$) determines the placements of the link intervals.

The remaining $3np' + 2xq'$ cut edges should come from the partial intersections of the link intervals with the E-gadgets, and the intersections among link intervals. We show that this is not possible if $y < x$. For this, consider an E-gadget $\mathcal{G}(v_i, v_j)$. Let $\ell_i, \ell'_i$ be the two link intervals from $\mathcal{G}(v_i)$ that intersect $\mathcal{G}(v_i, v_j)$ in the second manner and  $\ell_j, \ell'_j$ be the two link intervals from $\mathcal{G}(v_j)$ that intersect $\mathcal{G}(v_i, v_j)$ in the third manner. Consider the following cases: $\ell_i, \ell'_i$, $\ell_j, \ell'_j$ are in the same subset (Case 1), say $\ell_i, \ell'_i$, $\ell_j, \ell'_j \in A$  and $\ell_i, \ell'_i$ are in one subset and $\ell_j, \ell'_j$ are in the other (Case 2), say, $\ell_i, \ell'_i \in A$, $\ell_j, \ell'_j \in B$. In Case 2, the edge $(e_i, e_j)$ appears in the cut set of $C \cup D$, while in Case 1, it does not. For each case, we have two subcases as described in the following.

\begin{description}
 \item [Case 1a] $A$ contains $\ell_i, \ell'_i$, $\ell_j, \ell'_j$ and the left long and right short intervals of $\mathcal{G}(v_i, v_j)$. $B$ contains the right long and left short intervals of $\mathcal{G}(v_i, v_j)$. Hence, the intersections between $\mathcal{G}(v_i, v_j)$ and $\ell_i, \ell'_i$, $\ell_j, \ell'_j$ give $2p'$ cut edges.

 \item [Case 1b] $A$ contains $\ell_i, \ell'_i$, $\ell_j, \ell'_j$ and the right long and left short intervals of $\mathcal{G}(v_i, v_j)$. $B$ contains the left long and right short intervals of $\mathcal{G}(v_i, v_j)$. Hence, the intersections between $\mathcal{G}(v_i, v_j)$ and $\ell_i, \ell'_i$, $\ell_j, \ell'_j$ give $4q'$ cut edges.

 \item [Case 2a] $A$ contains $\ell_i, \ell'_i$ and the left long and right short intervals of $\mathcal{G}(v_i, v_j)$. $B$ contains $\ell_j, \ell'_j$ and the right long and left short intervals of $\mathcal{G}(v_i, v_j)$. Hence, the intersections between $\mathcal{G}(v_i, v_j)$ and $\ell_i, \ell'_i$, $\ell_j, \ell'_j$ give $2q'$ cut edges.

 \item [Case 2b] $A$ contains $\ell_i, \ell'_i$ and the right long and left short intervals of $\mathcal{G}(v_i, v_j)$. $B$ contains $\ell_j, \ell'_j$ and the left long and right short intervals of $\mathcal{G}(v_i, v_j)$. Hence, the intersections between $\mathcal{G}(v_i, v_j)$ and $\ell_i, \ell'_i$, $\ell_j, \ell'_j$ give $2p' + 2q'$ cut edges.  
 \end{description}


Therefore, we see that an  E-gadget gives at most $2p'$ cut edges from its partial intersections with link intervals if the link intervals belong to the same subset (since $2p' > 4q'$), and at most $2(p'+q')$ cut edges if the link intervals belong to different subsets (since $2p' + 2q' > 2q'$). Notice that the later case occurs for exactly $y$ E-gadgets. The number of cut edges obtained from the partial intersections of E-gadgets with link intervals is at most $2p'(\frac{3n}{2}-y) + 2(p'+ q')y = 3np' + 2yq'$. Hence if $y<x$, then at least $2(x-y)q' > 2q'$ cut edges must come from the intersections among the link intervals. But this is not possible as $2q' > 9n^2$. Hence $y \geq x$ as required.  
 \end{proof}

%
%

\begin{theorem}
 \textsc{Max Cut} is NP-complete on interval graphs.
\end{theorem}

\begin{proof}
 It can be checked in polynomial time if a given partition of an interval graph produces a cut of a given size. Thus the problem is in NP. The construction of $G'$ from $G$ clearly takes polynomial time. The NP-hardness follows from Lemma \ref{lem6}. 
\end{proof}

\section{Concluding Remarks}
 In this paper, we have settled the question of computational complexity of \textsc{Max Cut} on interval graphs. However, the question of whether \textsc{Max Cut} is polynomial-time solvable or NP-hard on unit interval graphs still remains open. For an NP-hardness reduction, a possible approach might be to reduce \textsc{Max Cut} on interval graphs to \textsc{Max Cut} on unit interval graphs. An interval can be transformed into a sequence of unit intervals by replacing it with a start and end interval, with ``bunches'' of unit intervals within (See Fig. \ref{fig: unit}). It is easy to see that for such a standalone gadget, an alternating assignment of the bunches to the two subsets yields a \textsc{Max Cut}. However, when multiple such gadgets of different sizes are brought together to represent the whole interval graph for the reduction, such a partition does not necessarily correspond to a partition in the original interval graph.

\begin{figure}[h!]
\centering

{
    \fontsize{8pt}{8pt}\selectfont
    \def\svgwidth{1\textwidth}
    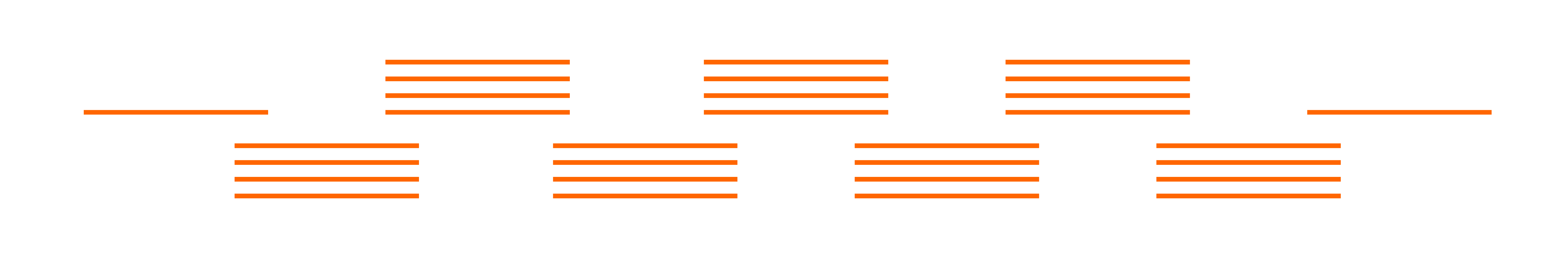
}

\caption[Short Caption]{Transformation of an interval into a sequence of unit intervals in a possible reduction from \textsc{Max Cut} on interval graphs to \textsc{Max Cut} on unit interval graphs.}
\label{fig: unit}
\end{figure}

Another direction for future work is to find approximation algorithms for \textsc{Max Cut} interval graphs. In general, polynomial-time approximation algorithm for \textsc{Max Cut} with the best known approximation ratio is by Goemans and Williamson \cite{goemans1995improved} which achieves an approximation ratio $\approx 0.878$. Assuming the Unique Games conjecture \cite{khot2002power}, this is the best possible approximation ratio. An interesting question is whether this can be bettered for interval graphs or unit interval graphs. A possible approach could be the following greedy method. We first compute a unit interval representation the graph. In the first step, the leftmost interval is put in $A$, then the leftmost interval not intersecting that interval is put in $A$, and so on. In the second step, among the remaining intervals, we consider the ones that intersect the most number of intervals put in $A$. The leftmost such interval is put in $B$, then the leftmost of them not intersecting that interval is put in $B$, and so on. We repeat this until all intervals are placed, i.e., in each odd (resp. even) step an independent set of intervals, each of which intersect the most number of intervals put in $B$ (resp $A$) thus far, are put in $A$ (resp. $B$). It is not clear to us how efficient this is, but the following is the worst example that we have found so far which gives an approximation ratio of 0.9375. Consider a graph $G = (V,E)$ with $8a$ vertices $\{v_1, \ldots, v_{8a}\}$ where the first $6a$ vertices are all adjacent to each other and the last $6a$ vertices are all adjacent to each other. The greedy algorithm gives a cut of size $15a^2$, while the maximum cut is of size $16a^2$.

\bibliographystyle{plainurl}
\bibliography{maxcut}

\end{document}

%% file: images/gadget1.pdf_tex
\begingroup%
  \makeatletter%
  \providecommand\color[2][]{%
    \errmessage{(Inkscape) Color is used for the text in Inkscape, but the package 'color.sty' is not loaded}%
    \renewcommand\color[2][]{}%
  }%
  \providecommand\transparent[1]{%
    \errmessage{(Inkscape) Transparency is used (non-zero) for the text in Inkscape, but the package 'transparent.sty' is not loaded}%
    \renewcommand\transparent[1]{}%
  }%
  \providecommand\rotatebox[2]{#2}%
  \ifx\svgwidth\undefined%
    \setlength{\unitlength}{1303.93700787bp}%
    \ifx\svgscale\undefined%
      \relax%
    \else%
      \setlength{\unitlength}{\unitlength * \real{\svgscale}}%
    \fi%
  \else%
    \setlength{\unitlength}{\svgwidth}%
  \fi%
  \global\let\svgwidth\undefined%
  \global\let\svgscale\undefined%
  \makeatother%
  \begin{picture}(1,0.52173913)%
    \put(0,0){\includegraphics[width=\unitlength,page=1]{gadget1.pdf}}%
    \put(0.22048987,0.50234419){\color[rgb]{0,0,0}\makebox(0,0)[lb]{\smash{$p$ left short}}}%
    \put(0.24774672,0.46500316){\color[rgb]{0,0,0}\makebox(0,0)[lb]{\smash{intervals}}}%
    \put(0.01981269,0.20099512){\color[rgb]{0,0,0}\rotatebox{91.06116476}{\makebox(0,0)[lb]{\smash{$q$ left long}}}}%
    \put(0.0566424,0.22893868){\color[rgb]{0,0,0}\rotatebox{91.06116476}{\makebox(0,0)[lb]{\smash{intervals}}}}%
    \put(0.9768184,0.1851586){\color[rgb]{0,0,0}\rotatebox{-89.15416554}{\makebox(0,0)[lb]{\smash{$q$ right long}}}}%
    \put(0.9398839,0.15735349){\color[rgb]{0,0,0}\rotatebox{-89.15416554}{\makebox(0,0)[lb]{\smash{intervals}}}}%
    \put(0,0){\includegraphics[width=\unitlength,page=2]{gadget1.pdf}}%
    \put(0.61437426,0.31828625){\color[rgb]{0,0,0}\makebox(0,0)[lb]{\smash{$p$ right short}}}%
    \put(0.64776637,0.28094523){\color[rgb]{0,0,0}\makebox(0,0)[lb]{\smash{intervals}}}%
    \put(0,0){\includegraphics[width=\unitlength,page=3]{gadget1.pdf}}%
  \end{picture}%
\endgroup%

%% file: images/gadget2.pdf_tex
\begingroup%
  \makeatletter%
  \providecommand\color[2][]{%
    \errmessage{(Inkscape) Color is used for the text in Inkscape, but the package 'color.sty' is not loaded}%
    \renewcommand\color[2][]{}%
  }%
  \providecommand\transparent[1]{%
    \errmessage{(Inkscape) Transparency is used (non-zero) for the text in Inkscape, but the package 'transparent.sty' is not loaded}%
    \renewcommand\transparent[1]{}%
  }%
  \providecommand\rotatebox[2]{#2}%
  \ifx\svgwidth\undefined%
    \setlength{\unitlength}{1303.93700787bp}%
    \ifx\svgscale\undefined%
      \relax%
    \else%
      \setlength{\unitlength}{\unitlength * \real{\svgscale}}%
    \fi%
  \else%
    \setlength{\unitlength}{\svgwidth}%
  \fi%
  \global\let\svgwidth\undefined%
  \global\let\svgscale\undefined%
  \makeatother%
  \begin{picture}(1,0.52173913)%
    \put(0,0){\includegraphics[width=\unitlength,page=1]{gadget2.pdf}}%
    \put(0.26343674,0.49743598){\color[rgb]{0,0,0}\makebox(0,0)[lb]{\smash{$p'$ left short}}}%
    \put(0.29437474,0.46009494){\color[rgb]{0,0,0}\makebox(0,0)[lb]{\smash{intervals}}}%
    \put(0.01893621,0.19696297){\color[rgb]{0,0,0}\rotatebox{91.06116476}{\makebox(0,0)[lb]{\smash{$q'$ left long}}}}%
    \put(0.05576602,0.22490654){\color[rgb]{0,0,0}\rotatebox{91.06116476}{\makebox(0,0)[lb]{\smash{intervals}}}}%
    \put(0.98032324,0.19041664){\color[rgb]{0,0,0}\rotatebox{-89.15416554}{\makebox(0,0)[lb]{\smash{$q'$ right long}}}}%
    \put(0.94338874,0.16261153){\color[rgb]{0,0,0}\rotatebox{-89.15416554}{\makebox(0,0)[lb]{\smash{intervals}}}}%
    \put(0,0){\includegraphics[width=\unitlength,page=2]{gadget2.pdf}}%
    \put(0.63523416,0.30969687){\color[rgb]{0,0,0}\makebox(0,0)[lb]{\smash{$p'$ right short}}}%
    \put(0.66249094,0.27235585){\color[rgb]{0,0,0}\makebox(0,0)[lb]{\smash{intervals}}}%
    \put(0,0){\includegraphics[width=\unitlength,page=3]{gadget2.pdf}}%
  \end{picture}%
\endgroup%

%% file: images/arrangement.pdf_tex
\begingroup%
  \makeatletter%
  \providecommand\color[2][]{%
    \errmessage{(Inkscape) Color is used for the text in Inkscape, but the package 'color.sty' is not loaded}%
    \renewcommand\color[2][]{}%
  }%
  \providecommand\transparent[1]{%
    \errmessage{(Inkscape) Transparency is used (non-zero) for the text in Inkscape, but the package 'transparent.sty' is not loaded}%
    \renewcommand\transparent[1]{}%
  }%
  \providecommand\rotatebox[2]{#2}%
  \ifx\svgwidth\undefined%
    \setlength{\unitlength}{3118.11023622bp}%
    \ifx\svgscale\undefined%
      \relax%
    \else%
      \setlength{\unitlength}{\unitlength * \real{\svgscale}}%
    \fi%
  \else%
    \setlength{\unitlength}{\svgwidth}%
  \fi%
  \global\let\svgwidth\undefined%
  \global\let\svgscale\undefined%
  \makeatother%
  \begin{picture}(1,0.27272727)%
    \put(0,0){\includegraphics[width=\unitlength,page=1]{arrangement.pdf}}%
    \put(0.00777231,0.05634202){\color[rgb]{0,0,0}\rotatebox{91.06116473}{\makebox(0,0)[lb]{\smash{Link}}}}%
    \put(0.02317387,0.04955505){\color[rgb]{0,0,0}\rotatebox{91.06116473}{\makebox(0,0)[lb]{\smash{intervals}}}}%
    \put(0.04499687,0.25163485){\color[rgb]{0,0,0}\makebox(0,0)[lb]{\smash{$\mathcal{G}(v_1)$}}}%
    \put(0.14810858,0.25151035){\color[rgb]{0,0,0}\makebox(0,0)[lb]{\smash{$\mathcal{G}(v_2)$}}}%
    \put(0.30302624,0.25117446){\color[rgb]{0,0,0}\makebox(0,0)[lb]{\smash{$\mathcal{G}(v_n)$}}}%
    \put(0.55487074,0.25103266){\color[rgb]{0,0,0}\makebox(0,0)[lb]{\smash{$\mathcal{G}(e_1)$}}}%
    \put(0.65702019,0.25100154){\color[rgb]{0,0,0}\makebox(0,0)[lb]{\smash{$\mathcal{G}(e_2)$}}}%
    \put(0.75994453,0.25115759){\color[rgb]{0,0,0}\makebox(0,0)[lb]{\smash{$\mathcal{G}(e_3)$}}}%
    \put(0.91321487,0.25163485){\color[rgb]{0,0,0}\makebox(0,0)[lb]{\smash{$\mathcal{G}(e_m)$}}}%
    \put(0,0){\includegraphics[width=\unitlength,page=2]{arrangement.pdf}}%
  \end{picture}%
\endgroup%

%% file: images/intersection0.pdf_tex
\begingroup%
  \makeatletter%
  \providecommand\color[2][]{%
    \errmessage{(Inkscape) Color is used for the text in Inkscape, but the package 'color.sty' is not loaded}%
    \renewcommand\color[2][]{}%
  }%
  \providecommand\transparent[1]{%
    \errmessage{(Inkscape) Transparency is used (non-zero) for the text in Inkscape, but the package 'transparent.sty' is not loaded}%
    \renewcommand\transparent[1]{}%
  }%
  \providecommand\rotatebox[2]{#2}%
  \ifx\svgwidth\undefined%
    \setlength{\unitlength}{1139.52755906bp}%
    \ifx\svgscale\undefined%
      \relax%
    \else%
      \setlength{\unitlength}{\unitlength * \real{\svgscale}}%
    \fi%
  \else%
    \setlength{\unitlength}{\svgwidth}%
  \fi%
  \global\let\svgwidth\undefined%
  \global\let\svgscale\undefined%
  \makeatother%
  \begin{picture}(1,0.74875622)%
    \put(0,0){\includegraphics[width=\unitlength,page=1]{intersection0.pdf}}%
  \end{picture}%
\endgroup%

%% file: images/intersection1.pdf_tex
\begingroup%
  \makeatletter%
  \providecommand\color[2][]{%
    \errmessage{(Inkscape) Color is used for the text in Inkscape, but the package 'color.sty' is not loaded}%
    \renewcommand\color[2][]{}%
  }%
  \providecommand\transparent[1]{%
    \errmessage{(Inkscape) Transparency is used (non-zero) for the text in Inkscape, but the package 'transparent.sty' is not loaded}%
    \renewcommand\transparent[1]{}%
  }%
  \providecommand\rotatebox[2]{#2}%
  \ifx\svgwidth\undefined%
    \setlength{\unitlength}{1139.52755906bp}%
    \ifx\svgscale\undefined%
      \relax%
    \else%
      \setlength{\unitlength}{\unitlength * \real{\svgscale}}%
    \fi%
  \else%
    \setlength{\unitlength}{\svgwidth}%
  \fi%
  \global\let\svgwidth\undefined%
  \global\let\svgscale\undefined%
  \makeatother%
  \begin{picture}(1,0.74875622)%
    \put(0,0){\includegraphics[width=\unitlength,page=1]{intersection1.pdf}}%
  \end{picture}%
\endgroup%

%% file: images/intersection2.pdf_tex
\begingroup%
  \makeatletter%
  \providecommand\color[2][]{%
    \errmessage{(Inkscape) Color is used for the text in Inkscape, but the package 'color.sty' is not loaded}%
    \renewcommand\color[2][]{}%
  }%
  \providecommand\transparent[1]{%
    \errmessage{(Inkscape) Transparency is used (non-zero) for the text in Inkscape, but the package 'transparent.sty' is not loaded}%
    \renewcommand\transparent[1]{}%
  }%
  \providecommand\rotatebox[2]{#2}%
  \ifx\svgwidth\undefined%
    \setlength{\unitlength}{1139.52755906bp}%
    \ifx\svgscale\undefined%
      \relax%
    \else%
      \setlength{\unitlength}{\unitlength * \real{\svgscale}}%
    \fi%
  \else%
    \setlength{\unitlength}{\svgwidth}%
  \fi%
  \global\let\svgwidth\undefined%
  \global\let\svgscale\undefined%
  \makeatother%
  \begin{picture}(1,0.74875622)%
    \put(0,0){\includegraphics[width=\unitlength,page=1]{intersection2.pdf}}%
  \end{picture}%
\endgroup%

%% file: images/intersection3.pdf_tex
\begingroup%
  \makeatletter%
  \providecommand\color[2][]{%
    \errmessage{(Inkscape) Color is used for the text in Inkscape, but the package 'color.sty' is not loaded}%
    \renewcommand\color[2][]{}%
  }%
  \providecommand\transparent[1]{%
    \errmessage{(Inkscape) Transparency is used (non-zero) for the text in Inkscape, but the package 'transparent.sty' is not loaded}%
    \renewcommand\transparent[1]{}%
  }%
  \providecommand\rotatebox[2]{#2}%
  \ifx\svgwidth\undefined%
    \setlength{\unitlength}{1139.52755906bp}%
    \ifx\svgscale\undefined%
      \relax%
    \else%
      \setlength{\unitlength}{\unitlength * \real{\svgscale}}%
    \fi%
  \else%
    \setlength{\unitlength}{\svgwidth}%
  \fi%
  \global\let\svgwidth\undefined%
  \global\let\svgscale\undefined%
  \makeatother%
  \begin{picture}(1,0.74875622)%
    \put(0,0){\includegraphics[width=\unitlength,page=1]{intersection3.pdf}}%
  \end{picture}%
\endgroup%

%% file: images/connection.pdf_tex
\begingroup%
  \makeatletter%
  \providecommand\color[2][]{%
    \errmessage{(Inkscape) Color is used for the text in Inkscape, but the package 'color.sty' is not loaded}%
    \renewcommand\color[2][]{}%
  }%
  \providecommand\transparent[1]{%
    \errmessage{(Inkscape) Transparency is used (non-zero) for the text in Inkscape, but the package 'transparent.sty' is not loaded}%
    \renewcommand\transparent[1]{}%
  }%
  \providecommand\rotatebox[2]{#2}%
  \ifx\svgwidth\undefined%
    \setlength{\unitlength}{3600bp}%
    \ifx\svgscale\undefined%
      \relax%
    \else%
      \setlength{\unitlength}{\unitlength * \real{\svgscale}}%
    \fi%
  \else%
    \setlength{\unitlength}{\svgwidth}%
  \fi%
  \global\let\svgwidth\undefined%
  \global\let\svgscale\undefined%
  \makeatother%
  \begin{picture}(1,0.23937008)%
    \put(0,0){\includegraphics[width=\unitlength,page=1]{connection.pdf}}%
    \put(0.13186256,0.20106255){\color[rgb]{0,0,0}\makebox(0,0)[lb]{\smash{$\mathcal{G}(v_i)$}}}%
    \put(0,0){\includegraphics[width=\unitlength,page=2]{connection.pdf}}%
    \put(0.45416004,0.20108485){\color[rgb]{0,0,0}\makebox(0,0)[lb]{\smash{$\mathcal{G}(v_j)$}}}%
    \put(0,0){\includegraphics[width=\unitlength,page=3]{connection.pdf}}%
    \put(0.83186294,0.20106255){\color[rgb]{0,0,0}\makebox(0,0)[lb]{\smash{$\mathcal{G}(v_i,v_j)$}}}%
    \put(0,0){\includegraphics[width=\unitlength,page=4]{connection.pdf}}%
  \end{picture}%
\endgroup%

%% file: images/unit.pdf_tex
\begingroup%
  \makeatletter%
  \providecommand\color[2][]{%
    \errmessage{(Inkscape) Color is used for the text in Inkscape, but the package 'color.sty' is not loaded}%
    \renewcommand\color[2][]{}%
  }%
  \providecommand\transparent[1]{%
    \errmessage{(Inkscape) Transparency is used (non-zero) for the text in Inkscape, but the package 'transparent.sty' is not loaded}%
    \renewcommand\transparent[1]{}%
  }%
  \providecommand\rotatebox[2]{#2}%
  \ifx\svgwidth\undefined%
    \setlength{\unitlength}{3741.73228346bp}%
    \ifx\svgscale\undefined%
      \relax%
    \else%
      \setlength{\unitlength}{\unitlength * \real{\svgscale}}%
    \fi%
  \else%
    \setlength{\unitlength}{\svgwidth}%
  \fi%
  \global\let\svgwidth\undefined%
  \global\let\svgscale\undefined%
  \makeatother%
  \begin{picture}(1,0.16666667)%
    \put(0,0){\includegraphics[width=\unitlength,page=1]{unit.pdf}}%
  \end{picture}%
\endgroup%

%% file: maxcut.bbl
\begin{thebibliography}{10}

\bibitem{asdre2007harmonious}
Katerina Asdre, Kyriaki Ioannidou, and Stavros~D. Nikolopoulos.
\newblock The harmonious coloring problem is {NP}-complete for interval and
  permutation graphs.
\newblock {\em Discret. Appl. Math.}, 155(17):2377--2382, 2007.
\newblock \href {https://doi.org/10.1016/j.dam.2007.07.005}
  {\path{doi:10.1016/j.dam.2007.07.005}}.

\bibitem{barahona1983max}
Francisco Barahona.
\newblock The max-cut problem on graphs not contractible to {$K_5$}.
\newblock {\em Operations Research Letters}, 2(3):107--111, 1983.
\newblock \href {https://doi.org/10.1016/0167-6377(83)90016-0}
  {\path{doi:10.1016/0167-6377(83)90016-0}}.

\bibitem{barahona1988application}
Francisco Barahona, Martin Gr{\"o}tschel, Michael J{\"u}nger, and Gerhard
  Reinelt.
\newblock An application of combinatorial optimization to statistical physics
  and circuit layout design.
\newblock {\em Operations Research}, 36(3):493--513, 1988.

\bibitem{berman1999some}
Piotr Berman and Marek Karpinski.
\newblock On some tighter inapproximability results.
\newblock In Jir{\'{\i}} Wiedermann, Peter van Emde~Boas, and Mogens Nielsen,
  editors, {\em Automata, Languages and Programming, 26th International
  Colloquium, ICALP'99, Prague, Czech Republic, July 11-15, 1999, Proceedings},
  volume 1644 of {\em Lecture Notes in Computer Science}, pages 200--209.
  Springer, 1999.
\newblock \href {https://doi.org/10.1007/3-540-48523-6\_17}
  {\path{doi:10.1007/3-540-48523-6\_17}}.

\bibitem{bodlaender1989achromatic}
Hans~L. Bodlaender.
\newblock Achromatic number is {NP}-complete for cographs and interval graphs.
\newblock {\em Inf. Process. Lett.}, 31(3):135--138, 1989.
\newblock \href {https://doi.org/10.1016/0020-0190(89)90221-4}
  {\path{doi:10.1016/0020-0190(89)90221-4}}.

\bibitem{bodlaender2004simple}
Hans~L. Bodlaender, Celina M.~H. de~Figueiredo, Marisa Gutierrez, Ton Kloks,
  and Rolf Niedermeier.
\newblock Simple max-cut for split-indifference graphs and graphs with few
  ${P_4}$s.
\newblock In Celso~C. Ribeiro and Simone~L. Martins, editors, {\em Experimental
  and Efficient Algorithms, Third International Workshop, {WEA} 2004, Angra dos
  Reis, Brazil, May 25-28, 2004, Proceedings}, volume 3059 of {\em Lecture
  Notes in Computer Science}, pages 87--99. Springer, 2004.
\newblock \href {https://doi.org/10.1007/978-3-540-24838-5\_7}
  {\path{doi:10.1007/978-3-540-24838-5\_7}}.

\bibitem{bodlaender2000complexity}
Hans~L Bodlaender and Klaus Jansen.
\newblock On the complexity of the maximum cut problem.
\newblock {\em Nordic Journal of Computing}, 7(1):14--31, 2000.

\bibitem{bodlaender1999simple}
Hans~L. Bodlaender, Ton Kloks, and Rolf Niedermeier.
\newblock {SIMPLE} {MAX-CUT} for unit interval graphs and graphs with few
  {$P_4$}s.
\newblock {\em Electron. Notes Discret. Math.}, 3:19--26, 1999.
\newblock \href {https://doi.org/10.1016/S1571-0653(05)80014-9}
  {\path{doi:10.1016/S1571-0653(05)80014-9}}.

\bibitem{boyaci2017polynomial}
Arman Boyaci, T{\'{\i}}naz Ekim, and Mordechai Shalom.
\newblock A polynomial-time algorithm for the maximum cardinality cut problem
  in proper interval graphs.
\newblock {\em Inf. Process. Lett.}, 121:29--33, 2017.
\newblock \href {https://doi.org/10.1016/j.ipl.2017.01.007}
  {\path{doi:10.1016/j.ipl.2017.01.007}}.

\bibitem{chakraborty2020algorithms}
Dibyayan Chakraborty, Sandip Das, Florent Foucaud, Harmender Gahlawat, Dimitri
  Lajou, and Bodhayan Roy.
\newblock Algorithms and complexity for geodetic sets on planar and chordal
  graphs.
\newblock In Yixin Cao, Siu{-}Wing Cheng, and Minming Li, editors, {\em 31st
  International Symposium on Algorithms and Computation, {ISAAC} 2020, December
  14-18, 2020, Hong Kong, China (Virtual Conference)}, volume 181 of {\em
  LIPIcs}, pages 7:1--7:15. Schloss Dagstuhl - Leibniz-Zentrum f{\"{u}}r
  Informatik, 2020.
\newblock \href {https://doi.org/10.4230/LIPIcs.ISAAC.2020.7}
  {\path{doi:10.4230/LIPIcs.ISAAC.2020.7}}.

\bibitem{chang1987}
K.~C. Chang and David~Hung{-}Chang Du.
\newblock Efficient algorithms for layer assignment problem.
\newblock {\em {IEEE} Trans. Comput. Aided Des. Integr. Circuits Syst.},
  6(1):67--78, 1987.
\newblock \href {https://doi.org/10.1109/TCAD.1987.1270247}
  {\path{doi:10.1109/TCAD.1987.1270247}}.

\bibitem{chang1998efficient}
Maw{-}Shang Chang.
\newblock Efficient algorithms for the domination problems on interval and
  circular-arc graphs.
\newblock {\em {SIAM} J. Comput.}, 27(6):1671--1694, 1998.
\newblock \href {https://doi.org/10.1137/S0097539792238431}
  {\path{doi:10.1137/S0097539792238431}}.

\bibitem{cohen1978food}
Joel~E Cohen and David~W Stephens.
\newblock {\em Food webs and niche space}.
\newblock Princeton University Press, 1978.

\bibitem{cohen2006optimal}
Johanne Cohen, Fedor~V. Fomin, Pinar Heggernes, Dieter Kratsch, and Gregory
  Kucherov.
\newblock Optimal linear arrangement of interval graphs.
\newblock In Rastislav Kralovic and Pawel Urzyczyn, editors, {\em 31st
  International Symposium on Mathematical Foundations of Computer Science,
  {MFCS} 2006, Star{\'{a}} Lesn{\'{a}}, Slovakia, August 28-September 1, 2006,
  Proceedings}, volume 4162 of {\em Lecture Notes in Computer Science}, pages
  267--279. Springer, 2006.
\newblock \href {https://doi.org/10.1007/11821069\_24}
  {\path{doi:10.1007/11821069\_24}}.

\bibitem{diaz2007max}
Josep D{\'{\i}}az and Marcin Kaminski.
\newblock {MAX-CUT} and {MAX-BISECTION} are {NP}-hard on unit disk graphs.
\newblock {\em Theor. Comput. Sci.}, 377(1-3):271--276, 2007.
\newblock \href {https://doi.org/10.1016/j.tcs.2007.02.013}
  {\path{doi:10.1016/j.tcs.2007.02.013}}.

\bibitem{foucaud2017}
Florent Foucaud, George~B. Mertzios, Reza Naserasr, Aline Parreau, and Petru
  Valicov.
\newblock Identification, location-domination and metric dimension on interval
  and permutation graphs. {II.} {A}lgorithms and complexity.
\newblock {\em Algorithmica}, 78(3):914--944, 2017.
\newblock \href {https://doi.org/10.1007/s00453-016-0184-1}
  {\path{doi:10.1007/s00453-016-0184-1}}.

\bibitem{garey1990guide}
Michael~R Garey and David~S Johnson.
\newblock A guide to the theory of {NP}-completeness.
\newblock {\em Computers and Intractability}, pages 37--79, 1990.

\bibitem{goemans1995improved}
Michel~X. Goemans and David~P. Williamson.
\newblock Improved approximation algorithms for maximum cut and satisfiability
  problems using semidefinite programming.
\newblock {\em J. {ACM}}, 42(6):1115--1145, 1995.
\newblock \href {https://doi.org/10.1145/227683.227684}
  {\path{doi:10.1145/227683.227684}}.

\bibitem{golumbic2004}
Martin~Charles Golumbic.
\newblock {\em Algorithmic {G}raph {T}heory and {P}erfect {G}raphs}.
\newblock Elsevier, 2004.

\bibitem{guruswami1999maximum}
Venkatesan Guruswami.
\newblock Maximum cut on line and total graphs.
\newblock {\em Discret. Appl. Math.}, 92(2-3):217--221, 1999.
\newblock \href {https://doi.org/10.1016/S0166-218X(99)00056-6}
  {\path{doi:10.1016/S0166-218X(99)00056-6}}.

\bibitem{hadlock1975finding}
F.~Hadlock.
\newblock Finding a maximum cut of a planar graph in polynomial time.
\newblock {\em {SIAM} J. Comput.}, 4(3):221--225, 1975.
\newblock \href {https://doi.org/10.1137/0204019} {\path{doi:10.1137/0204019}}.

\bibitem{imai1983finding}
Hiroshi Imai and Takao Asano.
\newblock Finding the connected components and a maximum clique of an
  intersection graph of rectangles in the plane.
\newblock {\em Journal of Algorithms}, 4(4):310--323, 1983.
\newblock \href {https://doi.org/10.1016/0196-6774(83)90012-3}
  {\path{doi:10.1016/0196-6774(83)90012-3}}.

\bibitem{JOHNSON1985434}
David~S Johnson.
\newblock The {NP}-completeness column: an ongoing guide.
\newblock {\em Journal of Algorithms}, 6(3):434 -- 451, 1985.
\newblock \href {https://doi.org/10.1016/0196-6774(85)90012-4}
  {\path{doi:10.1016/0196-6774(85)90012-4}}.

\bibitem{jungck1982computer}
John~R Jungck, Gregg Dick, and Amy~Gleason Dick.
\newblock Computer-assisted sequencing, interval graphs, and molecular
  evolution.
\newblock {\em Biosystems}, 15(3):259--273, 1982.
\newblock \href {https://doi.org/10.1016/0303-2647(82)90010-7}
  {\path{doi:10.1016/0303-2647(82)90010-7}}.

\bibitem{karp1972reducibility}
Richard~M. Karp.
\newblock Reducibility among combinatorial problems.
\newblock In Raymond~E. Miller and James~W. Thatcher, editors, {\em Proceedings
  of a symposium on the Complexity of Computer Computations, held March 20-22,
  1972, at the {IBM} Thomas J. Watson Research Center, Yorktown Heights, New
  York, {USA}}, The {IBM} Research Symposia Series, pages 85--103. Plenum
  Press, New York, 1972.
\newblock \href {https://doi.org/10.1007/978-1-4684-2001-2\_9}
  {\path{doi:10.1007/978-1-4684-2001-2\_9}}.

\bibitem{keil1985finding}
J.~Mark Keil.
\newblock Finding hamiltonian circuits in interval graphs.
\newblock {\em Inf. Process. Lett.}, 20(4):201--206, 1985.
\newblock \href {https://doi.org/10.1016/0020-0190(85)90050-X}
  {\path{doi:10.1016/0020-0190(85)90050-X}}.

\bibitem{khot2002power}
Subhash Khot.
\newblock On the power of unique 2-prover 1-round games.
\newblock In John~H. Reif, editor, {\em Proceedings on 34th Annual {ACM}
  Symposium on Theory of Computing, May 19-21, 2002, Montr{\'{e}}al,
  Qu{\'{e}}bec, Canada}, pages 767--775. {ACM}, 2002.
\newblock \href {https://doi.org/10.1145/509907.510017}
  {\path{doi:10.1145/509907.510017}}.

\bibitem{kratochvil2020u}
Jan Kratochv{\'{\i}}l, Tom{\'{a}}s Masar{\'{\i}}k, and Jana Novotn{\'{a}}.
\newblock U-bubble model for mixed unit interval graphs and its applications:
  The maxcut problem revisited.
\newblock In Javier Esparza and Daniel Kr{\'{a}}l', editors, {\em 45th
  International Symposium on Mathematical Foundations of Computer Science,
  {MFCS} 2020, August 24-28, 2020, Prague, Czech Republic}, volume 170 of {\em
  LIPIcs}, pages 57:1--57:14. Schloss Dagstuhl - Leibniz-Zentrum f{\"{u}}r
  Informatik, 2020.
\newblock \href {https://doi.org/10.4230/LIPIcs.MFCS.2020.57}
  {\path{doi:10.4230/LIPIcs.MFCS.2020.57}}.

\bibitem{lu1997linear}
Chin~Lung Lu and Chuan~Yi Tang.
\newblock A linear-time algorithm for the weighted feedback vertex problem on
  interval graphs.
\newblock {\em Inf. Process. Lett.}, 61(2):107--111, 1997.
\newblock \href {https://doi.org/10.1016/S0020-0190(96)00193-7}
  {\path{doi:10.1016/S0020-0190(96)00193-7}}.

\bibitem{marathe1992generalized}
Madhav~V. Marathe, R.~Ravi, and C.~Pandu Rangan.
\newblock Generalized vertex covering in interval graphs.
\newblock {\em Discret. Appl. Math.}, 39(1):87--93, 1992.
\newblock \href {https://doi.org/10.1016/0166-218X(92)90116-R}
  {\path{doi:10.1016/0166-218X(92)90116-R}}.

\bibitem{marx2005short}
D{\'{a}}niel Marx.
\newblock A short proof of the {NP}-completeness of minimum sum interval
  coloring.
\newblock {\em Oper. Res. Lett.}, 33(4):382--384, 2005.
\newblock \href {https://doi.org/10.1016/j.orl.2004.07.006}
  {\path{doi:10.1016/j.orl.2004.07.006}}.

\bibitem{zhang1994algorithm}
Peisen Zhang, Eric~A Schon, Stuart~G Fischer, Eftihia Cayanis, Janie Weiss,
  Susan Kistler, and Philip~E Bourne.
\newblock An algorithm based on graph theory for the assembly of contigs in
  physical mapping of {DNA}.
\newblock {\em Bioinformatics}, 10(3):309--317, 1994.
\newblock \href {https://doi.org/10.1093/bioinformatics/10.3.309}
  {\path{doi:10.1093/bioinformatics/10.3.309}}.

\end{thebibliography}
